\documentclass[twocolumn,american,aps,pra,showpacs]{revtex4}
\usepackage[T1]{fontenc}
\usepackage[latin9]{inputenc}
\usepackage{babel}
\usepackage{amsthm}
\usepackage{amsmath}
\usepackage{amssymb}
\usepackage[unicode=true,pdfusetitle,
 bookmarks=true,bookmarksnumbered=false,bookmarksopen=false,
 breaklinks=false,pdfborder={0 0 0},backref=false,colorlinks=false]
 {hyperref}

\makeatletter
\@ifundefined{textcolor}{}
{%
 \definecolor{BLACK}{gray}{0}
 \definecolor{WHITE}{gray}{1}
 \definecolor{RED}{rgb}{1,0,0}
 \definecolor{GREEN}{rgb}{0,1,0}
 \definecolor{BLUE}{rgb}{0,0,1}
 \definecolor{CYAN}{cmyk}{1,0,0,0}
 \definecolor{MAGENTA}{cmyk}{0,1,0,0}
 \definecolor{YELLOW}{cmyk}{0,0,1,0}
}
\theoremstyle{plain}
\newtheorem{thm}{\protect\theoremname}

\usepackage{babel}
\usepackage{babel}
\usepackage{times}\usepackage{txfonts}\usepackage{braket}\usepackage[usenames,dvipsnames]{pstricks}\usepackage{epsfig}\usepackage{pst-grad}

\makeatother

\providecommand{\theoremname}{Theorem}

\begin{document}

\title{Quantifying nonclassicality: global impact of local unitary evolutions}

\author{S. M. Giampaolo$^{2}$, A. Streltsov$^{1}$, W. Roga$^{2}$, D. Bru{ß}$^{1}$,
and F. Illuminati$^{2,3}$}

\affiliation{ $^{1}$\mbox{Heinrich-Heine-Universität Düsseldorf, Institut für
Theoretische Physik III, D-40225 Düsseldorf, Germany} \\
 $^{2}$\mbox{Dipartimento di Ingegneria Industriale, Universit\`a
degli Studi di Salerno,~Via Ponte don Melillo, I-84084 Fisciano (SA),
Italy} \\
 $^{3}$Corresponding Author: illuminati@sa.infn.it}

\date{\today}

\begin{abstract}
We show that only those composite quantum systems possessing nonvanishing
quantum correlations have the property that \emph{any} nontrivial
local unitary evolution changes their global state. We derive the
exact relation between the global state change induced by local unitary
evolutions and the amount of quantum correlations. We prove that the
minimal change coincides with the geometric measure of discord (defined
via the Hilbert-Schmidt norm), thus providing the latter with an operational
interpretation in terms of the capability of a local unitary dynamics
to modify a global state. We establish that two-qubit Werner states
are maximally quantum correlated, and are thus the ones that maximize
this type of global quantum effect. Finally, we show that similar results
hold when replacing the Hilbert-Schmidt norm with the trace norm.
\end{abstract}


\maketitle
Although the existence of quantum correlations more general than entanglement
has been known for some time~\cite{Ollivier2001,Henderson2001,Oppenheim2002},
they have begun to attract increasing interest only after the recent
suggestion that they might constitute key resources for quantum information
and computation tasks, such as the computational speed-up in the model
of deterministic quantum computation with one pure qubit (DQC1)~\cite{Knill1998}.
In this model, use of a mixed separable state appears to allow for
the efficient, i.e. polynomial time, computation of the trace of any
$n$-qubit unitary matrix~\cite{Lanyon2008}, a problem believed
to fall in the $NP$ class on a classical computer~\cite{Laflamme2002,Datta2005}.
Given the absence of entanglement, and assuming the essential nonclassicality
of the protocol, this has led to suggest that a particular measure
of bipartite quantum correlations, the quantum discord \cite{Ollivier2001},
is the figure of merit for quantum computation with mixed states~\cite{Datta2008}.
Despite much progress, the issue is however not yet conclusively settled~\cite{Dakic2010,Merali2011,Chaves2011,Dakic2012}.
More recently, various operational interpretations of the quantum
discord and other measures of quantum correlations have been established~\cite{Ferraro2010,Madhok2011,Cavalcanti2011,Streltsov2011,Piani2011,Merali2011,Girolami2012,Streltsov2012,Chuan2012}.
Quantum discord in its entropic definition, i.e. as the difference
between two classically equivalent forms of mutual information \cite{Ollivier2001},
has been given its first information-theoretic operational meaning
in terms of entanglement consumption in an extended quantum-state-merging
protocol. Its asymmetry, i.e. the fact that in general the discord
between parties $A$ and $B$ given that party $A$ is measured is
different from the one given that party $B$ is measured, has been
related to the performance imbalance in quantum state merging and
dense coding~\cite{Cavalcanti2011}. The quantum discord has also
been shown to be equal to the minimal partial distillable entanglement,
that is the part of entanglement which is lost when one ignores the
subsystem which is not measured in a local projective measurement~\cite{Streltsov2011}.
Finally, a different measure of nonclassicality, the relative entropy
of quantumness, has been shown to be equivalent to the minimum distillable
entanglement generated between a system and local ancillae in a suitably
devised activation protocol~\cite{Piani2011}.

Notwithstanding these recent progresses, several fundamental questions
on the nature and properties of quantum correlations are yet to be
addressed. Among them a conceptually appealing one is determining
a unified mathematical framework for the quantification of entanglement
and quantumness. Such framework would allow to devise a basic physical
interpretation of quantum correlations and formulate sharp quantitative
questions on the ensuing measure of nonclassicality, such as the definition
and properties of maximally quantum-correlated states. In the present
work we define a distance-based measure of quantumness that for pure
states reduces to a particular distance-based measure of entanglement,
the so-called ``stellar entanglement'' \cite{Giampaolo2007,Monras2011}.
The latter associates pure-state bipartite entanglement to the minimal
change of a state induced by local unitary operations. It is a \textit{bona
fide} entanglement monotone for $M\times N$-dimensional composite
quantum systems and extends to mixed states via the convex roof construction.
Indeed, the research program on the global effects of local unitary
operations acting on composite quantum systems has turned out to be
fruitful in the investigation of various other issues \cite{Gharibian2009,Fu2006},
including the quantification of measurement-induced nonlocality \cite{Luo2011}
and the theory and applications of ground-state factorization in the
study of complex quantum systems~\cite{Giampaolo2008,Giampaolo2009,Giampaolo2010}.
Very recently, the possibility of quantifying quantum correlations
via the effect of local unitary operations has been discussed in Ref.~\cite{Gharibian2012}.

In the present work we shall show that the minimal disturbance on
mixed bipartite quantum states under the action of local unitary (Hamiltonian)
time evolutions on only one of the parties defines a faithful measure
of quantum correlations vanishing if and only if the state is classically
correlated and reducing to the stellar entanglement for pure states.
This measure enjoys a clear physical interpretation in terms of the
impact power of local unitary time-evolutions, i.e. the ability to
induce a global state change. Moreover, at least for two-qubit systems,
it coincides with the geometric measure of discord defined as the
distance from the set of classically correlated states using the Hilbert-Schmidt
norm \cite{Dakic2010}. In the case of two-qubit systems and for any
value of the global state purity, we find that the measure is maximized
by the class of two-qubit Werner states. Furthermore, for the general
case of $m\times n$-dimensional systems, we show that the impact
power is an upper bound to the geometric discord. Finally, we will
briefly comment on the extension of the present investigation when
the Hilbert-Schmidt norm is replaced by other norms.




Let us begin by considering a bipartite quantum system composed by
two subsystems, $A$ and $B$, so that the Hilbert space ${\cal H}={\cal H}_{A}\otimes{\cal H}_{B}$.
Under the evolution driven by a local Hamiltonian $H_{A}$ acting
only subsystem $A$ the global density matrix $\rho^{AB}$ evolves
accordingly to the unitary Schrödinger dynamics:
\begin{equation}
\rho^{AB}\left(t\right)=e^{-iH_{A}t}\rho^{AB}e^{iH_{A}t}\;.\label{eq:rho}
\end{equation}
In order to quantify the effect of such a local unitary time-evolution
on any given initial global state we define the {\em impact} of
the Hamiltonian $H_{A}$ as the Hilbert-Schmidt distance between the
evolved state at time $t$ and the initial state:
\begin{equation}
I\left(\rho^{AB},H_{A},t\right)=\frac{1}{2}\left\Vert \rho^{AB}\left(t\right)-\rho^{AB}\right\Vert ^{2}\;,\label{eq:I}
\end{equation}
where $\left\Vert \rho-\sigma\right\Vert ^{2}=\mathrm{Tr}[(\rho-\sigma)^{2}]$
is the Hilbert-Schmidt distance. The impact vanishes if the time evolution
does not affect the initial state as in the trivial cases in which
either $t=0$ or $H_{A}\propto\openone_{A}$. On the other hand it
can never exceed unity, as it can be seen by noticing that for any
two arbitrarily chosen quantum states $\rho$ and $\gamma$ one has
$\frac{1}{2}\left\Vert \rho-\gamma\right\Vert ^{2}=\frac{1}{2}\left(\mathrm{Tr}\left[\rho^{2}\right]+\mathrm{Tr}\left[\gamma^{2}\right]-2\mathrm{Tr}\left[\rho\gamma\right]\right)\leq\frac{1}{2}\left(\mathrm{Tr}\left[\rho^{2}\right]+\mathrm{Tr}\left[\gamma^{2}\right]\right)\le1$.
The above inequality also implies that the impact reaches unity if
and only if the time evolution driven by $H_{A}$ takes an initial
pure state into another pure state orthogonal to it.

Given the Hamiltonian $H_{A}$ and the initial state $\rho^{AB}$,
we aim to determine the maximum possible value of the impact $I$
with respect to time $t$. Hence, we introduce the \emph{impact power}
$P$ of a Hamiltonian $H_{A}$ with respect to the initial state $\rho^{AB}$:
\begin{equation}
P\left(\rho^{AB},H_{A}\right)=\max_{t}I\left(\rho^{AB},H_{A},t\right)\;.
\end{equation}
If $H_{A}$ is trivial, i.e. $H_{A}\propto\openone_{A}$, then $P\left(\rho^{AB},H_{A}\right)\equiv0$.
Let us consider the case in which $A$ is a qubit while $B$ can be
any $d$-dimensional system. Any nontrivial local Hamiltonians $H_{A}$
can then be written as $H_{A}=E_{0}\Pi_{0}^{A}+E_{1}\Pi_{1}^{A}$
where $E_{0}\neq E_{1}$ are the two nondegenerate energy eigenvalues
and $\Pi_{i}^{A}$ are the orthogonal projectors onto the two energy
eigenstates $\ket{0}$ and $\ket{1}$. With this expression of $H_{A}$
the impact power reads
\begin{equation}
P\left(\rho^{AB},H_{A}\right)=\max_{t}\left\{ a-b\cos\left(\Delta Et\right)\right\} \;,\label{eq:P2}
\end{equation}
where the energy gap $\Delta E=E_{1}-E_{0}$ and the time-independent
quantities $a$ and $b$ are
\begin{eqnarray}
a & = & \mathrm{Tr}\left[\left(\rho^{AB}\right)^{2}\right]-\mathrm{Tr}\left[\rho^{AB}\sum_{i=0}^{1}\Pi_{i}^{A}\rho^{AB}\Pi_{i}^{A}\right]\;;\label{eq:a}\\
b & = & 2\mathrm{Tr}\left[\rho^{AB}\Pi_{1}^{A}\rho^{AB}\Pi_{0}^{A}\right]\;.\label{eq:b}
\end{eqnarray}
Notice that $b$ is nonnegative, since it can be written as $2\mathrm{Tr}\left[XX^{\dagger}\right]$
with $X=\Pi_{0}^{A}\rho^{AB}\Pi_{1}^{A}$. The fact that $a$ and
$b$ are constants and $b\geq0$ implies that the impact reaches its
maximum $a+b$ at times $t_{\mathrm{max}}^{(k)}=\frac{(2k+1)\pi}{\Delta E}$,
with $k$ integer. Exploiting completeness, $\sum_{i}\Pi_{i}^{A}=\openone_{A}$,
one has $\mathrm{Tr}\left[\left(\rho^{AB}\right)^{2}\right]=\mathrm{Tr}\left[\rho^{AB}\left(\Pi_{0}^{A}+\Pi_{1}^{A}\right)\rho^{AB}\left(\Pi_{0}^{A}+\Pi_{1}^{A}\right)\right]$.
As a consequence, $a=b$. Indeed, this result can be
obtained straightforwardly from Eq. \ref{eq:P2} by setting $t=0$ and reminding that
at $t=0$ it must be $P=0$. Exploiting the equality $a=b$, we then have:
\begin{equation}
P\left(\rho^{AB},H_{A}\right)\!=\!2\left\{ \mathrm{Tr}\left[\left(\rho^{AB}\right)^{2}\right]-\mathrm{Tr}\left[\rho^{AB}\sum_{i=0}^{1}\Pi_{i}^{A}\rho^{AB}\Pi_{i}^{A}\right]\right\} \,.\label{Impact_power}
\end{equation}
The impact power $P$ cannot exceed unity and one has strictly $P<1$
if the initial state is mixed. Maximizing over all $H_{A}$ we can
define the maximal possible impact power for any given initial state
$\rho^{AB}$ as $P_{\max}\left(\rho^{AB}\right)\!=\!\max_{H_{A}}\! P\left(\rho^{AB},H_{A}\right)$.
From this definition it follows immediately that $P_{\max}\left(\rho^{AB}\right)<1$
for all mixed states. On the other hand, 
it is known that an initial \textit{pure} state is a product state
if and only if there exists at least one local unitary traceless operation
that leaves it invariant~\cite{Giampaolo2007,Monras2011}. For any given initial
state $\rho^{AB}$ we can then introduce the smallest possible impact power
$P_{\min}\left(\rho^{AB}\right)$, defined by minimizing $P$ over all local
Hamiltonians that are not proportional to the identity:
\begin{equation}
P_{\min}\left(\rho^{AB}\right) =
\underset{H_{A}\neq\alpha\openone_{A}}{\min}P(\rho^{AB},H_{A}) \; .
\end{equation}
It is evident from the definition that $P_{min}\left(\rho^{AB}\right)$
vanishes if and only if $\rho^{AB}$ is a product pure state. For entangled
pure states $P_{\min}\left(\rho^{AB}\right)$ cannot vanish because,
due to the presence of the entanglement, any local perturbation acting
on a subsystem will affect the entire system. Starting from this result,
when we move from the case of pure entangles states to that of mixed
nonclassical states we find a similar behavior, but for the important
difference that the role previously played by the entanglement is now
played by the quantum correlations. Indeed, we will now show that $P_{\min}\left(\rho^{AB}\right)$
is directly related to a well defined measure of bipartite quantum
correlations, that is, the 
geometric measure of discord $D_{A}^{\left(2\right)}\left(\rho^{AB}\right)$~\cite{Dakic2010},
defined as:
\begin{equation}
D_{A}^{\left(2\right)}\left(\rho^{AB}\right)=\min_{\omega^{AB}\in CQ}\left\Vert \rho^{AB}-\omega^{AB}\right\Vert ^{2}\;.
\end{equation}
In the definition of the geometric discord the minimization is taken
over the set $CQ$ of all classically correlated states,
that is states of the form $\omega^{AB}=\sum_{i}p_{i}\ket{i}\bra{i}^{A}\otimes\omega_{i}^{B}$
where $\omega_{i}^{B}$ is a state on subsystem $B$. Using Eq.~(\ref{Impact_power})
together with the equality $\mathrm{Tr}[\rho^{AB}\sum_{i=0}^{1}\Pi_{i}^{A}\rho^{AB}\Pi_{i}^{A}]=\mathrm{Tr}[(\sum_{i=0}^{1}\Pi_{i}^{A}\rho^{AB}\Pi_{i}^{A})^{2}]$
one can immediately verify by inspection that for any nondegenerate
single-qubit Hamiltonian $H_{A}=E_{0}\Pi_{0}^{A}+E_{1}\Pi_{1}^{A}$
the impact power can be written as $P\left(\rho^{AB},H_{A}\right)=2\left\Vert \rho^{AB}-\sum_{i=0}^{1}\Pi_{i}^{A}\rho^{AB}\Pi_{i}^{A}\right\Vert ^{2}$.
This implies the following order relation between the impact power
and the geometric measure of discord:
\begin{equation}
P\left(\rho^{AB},H_{A}\right)\ge2D_{A}^{\left(2\right)}\left(\rho^{AB}\right)\;.\label{Impact_power2}
\end{equation}

\begin{figure}[t]
\includegraphics[width=8cm]{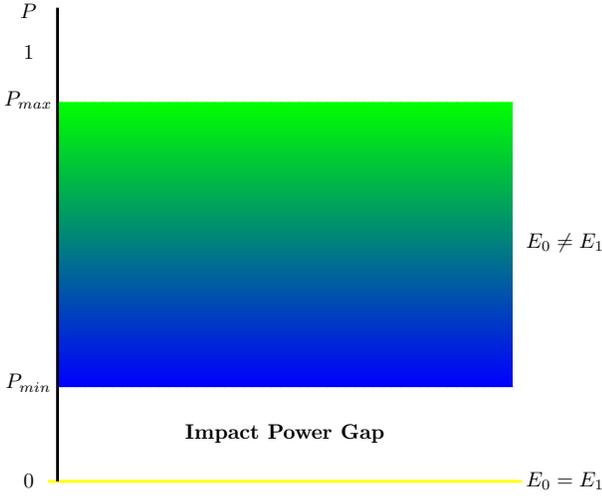}
\caption{\label{fig:potential gap} Possible values of the impact power $P$
for an arbitrary initial state $\rho^{AB}$. The impact power is zero
if the spectrum of the local Hamiltonian $H_{A}$ is degenerate: $E_{0}=E_{1}$
(yellow line). For $E_{0}\neq E_{1}$ the impact power can only take
values between $P_{\min}$ and $P_{\mathrm{max}}$ (green-blue area).
The impact power gap is the region between $0$ and $P_{\min}$. Its
width is measured by the amount of quantum correlations present in
the initial state $\rho^{AB}$, as measured by the geometric measure
of discord: $P_{\min}=2D_{A}^{\left(2\right)}$. See main text for
details.}
\label{fig1}
\end{figure}

Eq.~(\ref{Impact_power2}) shows that the change in the global state
due to a local unitary dynamics is bounded from below by the geometric
measure of discord and hence cannot vanish in the presence of quantum
correlations. Actually, one can prove a much stronger relation between
the minimum impact power $P_{min}$, that from now will be named the
{\em impact power gap}, and the geometric measure of
discord according to the following theorem:


\begin{thm}
\label{thm:1} If $\rho^{AB}$ is a state of a bipartite system, where
subsystem $A$ is a qubit, then the impact power gap $P_{min}$ is
given by:
\begin{equation}
P_{\min}\left(\rho^{AB}\right)=2D_{A}^{\left(2\right)}\left(\rho^{AB}\right)\;.\label{eq:gap}
\end{equation}

\end{thm}
\begin{proof} We will prove this equality by identifying a Hamiltonian
which explicitly minimizes the impact power $P\left(\rho^{AB},H_{A}\right)$.
To this end, it is useful to recall that the geometric measure of
discord is related to local von Neumann measurements, with local projectors
$\Pi_{i}^{A}$, according to the following~\cite{Luo2010}:
\begin{equation}
D_{A}^{\left(2\right)}\left(\rho^{AB}\right)=\min_{\left\{ \Pi_{i}^{A}\right\} }\left\Vert \rho^{AB}-\sum_{i}\Pi_{i}^{A}\rho^{AB}\Pi_{i}^{A}\right\Vert ^{2}\;.\label{LuoRelation}
\end{equation}
Let now $\hat{\Pi}_{0}^{A}$ and $\hat{\Pi}_{1}^{A}$ be the projectors
that achieve the minimum and consider the Hamiltonian $H_{A}=E_{0}\hat{\Pi}_{0}^{A}+E_{1}\hat{\Pi}_{1}^{A}$
with nondegenerate spectrum $E_{1}\neq E_{0}$. Evaluating the impact
power of $H_{A}$ along the same lines discussed in the cases above
yields $P\left(\rho^{AB},H_{A}\right)=2D_{A}^{\left(2\right)}\left(\rho^{AB}\right)$.
\end{proof}

Theorem~\ref{thm:1} exemplifies the relation between the impact
power gap and quantum correlations (see also Fig.~\ref{fig:potential gap}).
If subsystem $A$ is a qubit, then $P_{\min}$ can be computed explicitly
by exploiting Theorem~\ref{thm:1} and the explicit expression for
$D_{A}^{\left(2\right)}$ provided in Refs.~\cite{Dakic2010,Vinjanampathy2012}.
In fact, we can go one step further and provide independent closed
expressions both for $P_{\mathrm{min}}$ and for the maximal impact
power $P_{\mathrm{max}}$ in terms of the global state \textit{purity}:
\begin{thm} \label{thm:3} If system $A$ is a qubit, the maximal
impact power $P_{\mathrm{max}}$ reads
\begin{eqnarray}
P_{\mathrm{max}}\left(\rho^{AB}\right) & = & \mathrm{Tr}\left[\left(\rho^{AB}\right)^{2}\right]-m_{\mathrm{min}}\;,\label{eq:Pmax-1}
\end{eqnarray}
where $m_{\mathrm{min}}$ is the smallest eigenvalue of the matrix
$M$ with elements $M_{ij}=\mathrm{Tr}\left[\rho^{AB}\sigma_{i}^{A}\rho^{AB}\sigma_{j}^{A}\right]$,
where $\sigma_{i}^{A}$ (with $i=x,\, y,\, z$) are the Pauli operators
of subsystem $A$. Moreover, given the largest eigenvalue $m_{\mathrm{max}}$
of the matrix $M$, the impact power gap $P_{\mathrm{min}}$ reads
\begin{eqnarray}
P_{\mathrm{min}}\left(\rho^{AB}\right) & = & \mathrm{Tr}\left[\left(\rho^{AB}\right)^{2}\right]-m_{\mathrm{max}}\;.\label{eq:Pmin-1}
\end{eqnarray}
\end{thm} \begin{proof} Since the impact power is identically vanishing if the single-qubit Hamiltonian $H_A$ is degenerate, we need consider only the nondegenerate case. The unitary operator $U_A=e^{iH_At_{\rm max}^{(0)}}$ is then traceless with spectrum composed by the two complex roots of the unity. Let us recall Eq.~(\ref{Impact_power}) for the impact power
$P\left(\rho^{AB},H_{A}\right)$ and the fact that we can always rewrite
a local unitary operator in the form $U_{A}=\Pi_{0}^{A}-\Pi_{1}^{A}$.
We can then express the impact power as follows:
\begin{equation}
P\left(\rho^{AB},H_{A}\right)=\mathrm{Tr}\left[\left(\rho^{AB}\right)^{2}\right]-\mathrm{Tr}\left[\rho^{AB}U_{A}\rho^{AB}U_{A}^{\dagger}\right].\label{eq:P-1-1}
\end{equation}
Using the Bloch representation to write the projectors as $\Pi_{0}^{A}=\frac{1}{2}\left(\openone_{A}+\sum_{i}r_{i}\sigma_{i}^{A}\right)$
and $\Pi_{1}^{A}=\frac{1}{2}\left(\openone_{A}-\sum_{i}r_{i}\sigma_{i}^{A}\right)$,
the unitary operator $U_{A}$ in Eq.~(\ref{eq:P-1-1}) takes the
form $U_{A}=\Pi_{0}^{A}-\Pi_{1}^{A}=\sum_{i}r_{i}\sigma_{i}^{A}$.
The final expression for the impact power becomes
\begin{eqnarray}
P\left(\rho^{AB},H_{A}\right) & = & \mathrm{Tr}\left[\left(\rho^{AB}\right)^{2}\right]-\sum_{i,j}r_{i}M_{ij}r_{j}\;,\label{eq:P-2-1}
\end{eqnarray}
where we defined the matrix $M$ with the elements $M_{ij}=\mathrm{Tr}\left[\rho^{AB}\sigma_{i}^{A}\rho^{AB}\sigma_{j}^{A}\right]$.
It is easy to see that $M$ is symmetric, since $M_{ij}=M_{ji}$.
Moreover, all entries of $M$ are real. This implies that in order
to compute $P_{\max}$ we have to minimize $\boldsymbol{r}^{T}M\boldsymbol{r}$
over all unit vectors $\boldsymbol{r}$ for a real symmetric matrix
$M$. This problem is solved by finding the smallest eigenvalue of
$M$~\cite{Horn1990}. The impact power gap $P_{\min}$ can be computed
similarly by considering the largest eigenvalue of $M$. \end{proof}
By continuity in the Bloch vector $\boldsymbol{r}$, the impact power
$P\left(\rho^{AB},H_{A}\right)$ may assume any real value in the
range $[P_{\min},P_{\max}]$.

Equipped with these results, we can look for the class of states that,
at fixed global purity, maximize the impact power gap and thus the
quantum correlations. When both subsystems are qubits ($d_{A}=d_{B}=2$),
the following theorem holds: \begin{thm} For any state $\rho^{AB}$
of two qubits
\begin{equation}
P_{\min}\left(\rho^{AB}\right)\leq\frac{4}{3}\mathrm{Tr}\left[\left(\rho^{AB}\right)^{2}\right]-\frac{1}{3}\;,\label{eq:bound}
\end{equation}
with equality achieved by the Werner states $\rho_{w}$. \end{thm}
\begin{proof} In the Bloch sphere representation any arbitrary two-qubit
state can be written as:
\begin{equation}
\rho^{AB}\!\!=\!\!\frac{1}{4}\left(\!\!\openone\!\otimes\!\openone\!+\!\!\sum_{i}x_{i}\sigma_{i}\!\otimes\!\openone\!+\!\!\sum_{i}y_{i}\!\openone\!\otimes\!\sigma_{i}\!+\!\!\sum_{ij}T_{ij}\sigma_{i}\!\otimes\!\sigma_{j}\right)\,,\label{general_two_qubit}
\end{equation}
and the state purity $\mathrm{Tr}\left[(\rho^{AB})^{2}\right]$ can
be expressed as $\mathrm{Tr}\left[(\rho^{AB})^{2}\right]=\frac{1}{4}\left(1+\boldsymbol{x}^{2}+\boldsymbol{y}^{2}+\left\Vert T\right\Vert ^{2}\right)$.
By tracing out the first or the second qubit, the purities of the
reduced states are, respectively, $\mathrm{Tr}\left[(\rho^{B})^{2}\right]=\frac{1}{2}\left(1+\boldsymbol{y}^{2}\right)$
and $\mathrm{Tr}\left[(\rho^{A})^{2}\right]=\frac{1}{2}\left(1+\boldsymbol{x}^{2}\right)$.
Using representation Eq.~(\ref{general_two_qubit}), it is possible
to evaluate the geometric measure of discord for any two-qubit state~\cite{Dakic2010},
and hence the expression for $P_{\min}$:
\begin{equation}
P_{\min}\left(\rho^{AB}\right)=\frac{1}{2}\left(\boldsymbol{x}^{2}+\left\Vert T\right\Vert ^{2}-k_{\max}\right)\;,\label{eq:D-1}
\end{equation}
where $k_{\max}$ is the largest eigenvalue of the matrix $K=\boldsymbol{x}\boldsymbol{x}^{T}+TT^{T}$,
and $\left\Vert T\right\Vert ^{2}=\mathrm{Tr}\left[T^{T}T\right]$.
Since $k_{\max}$ is the largest eigenvalue of the $3\times3$ matrix
$K$, we have that $3k_{\max}\geq\boldsymbol{x}^{2}+\left\Vert T\right\Vert ^{2}$.
Using this inequality in Eq.~(\ref{eq:D-1}) and taking into account
the expressions of the global and reduced purities, we have:
\begin{eqnarray}
P_{\min}\left(\rho^{AB}\right) & \leq & \frac{1}{3}\left(\boldsymbol{x}^{2}+\left\Vert T\right\Vert ^{2}\right)\nonumber \\
 & = & \frac{4}{3}\left(\mathrm{Tr}\left[\left(\rho^{AB}\right)^{2}\right]-\frac{1}{2}\mathrm{Tr}\left[\left(\rho^{B}\right)^{2}\right]\right)\;.
\end{eqnarray}
Finally, noticing that for a single-qubit state the purity cannot
be smaller than $\frac{1}{2}$, we arrive at Ineq.~(\ref{eq:bound}).
On the other hand, a generic two-qubit Werner state can be written
as $\rho_{w}=\frac{2-x}{6}\openone+\frac{2x-1}{6}F$ where $x\in\left[-1,1\right]$
and $F=\sum_{k,l}\ket{k}\bra{l}\otimes\ket{l}\bra{k}$ is the permutation
operator. For such a state the purity is given by $\mathrm{Tr}\left[\rho_{w}^{2}\right]=\frac{1}{3}\left(x^{2}-x+1\right)$,
while the geometric measure of discord reads~\cite{Luo2010}: $D_{A}^{\left(2\right)}\left(\rho_{w}\right)=\frac{\left(2x-1\right)^{2}}{18}$.
Recalling the relation between the impact power gap and the geometric
discord, one has that Ineq.~(\ref{eq:bound}) is saturated by the
Werner states. Werner states are thus maximally quantum-correlated
two-qubit states at fixed global purity. \end{proof}

We could not yet clarify whether the Werner states are the only one
maximizing the two-qubit quantum correlations. Some preliminary analysis
suggests that other classes of highly symmetric states, like the isotropic
states, might also saturate the bound Eq. \ref{eq:bound}.

In order to investigate systems with larger local dimension $d_{A}>2$,
we generalize our approach considering the fully nondegenerate local
Hamiltonians of the form $H_{A}=\sum_{i=0}^{d_{A}-1}E_{i}\Pi_{i}^{A}$
with spectrum $E_{i}\neq E_{j}\,\forall\, i\neq j$. Following the
same route of reasoning as in the qubit case, we find that the impact
power of $H_{A}$ over an arbitrary initial state $\rho^{AB}$ can
be expressed as
\begin{equation}
P\left(\rho^{AB},H_{A}\right)=\max_{t}\left\{ a-\sum_{l>k}b_{lk}\cdot\cos\left(\Delta E_{lk}t\right)\right\} \;,\label{eq:P-4}
\end{equation}
where $\Delta E_{lk}=E_{l}-E_{k}$, and the coefficients $a$ and
$b_{lk}$ are
\begin{eqnarray}
a & = & \mathrm{Tr}\left[(\rho^{AB})^{2}\right]-\mathrm{Tr}\left[\rho^{AB}\sum_{i=0}^{d_{A}-1}\Pi_{i}^{A}\rho^{AB}\Pi_{i}^{A}\right]\;;\label{eq:a:2}\\
b_{lk} & = & 2\mathrm{Tr}\left[\rho^{AB}\Pi_{l}^{A}\rho^{AB}\Pi_{k}^{A}\right]\;.
\end{eqnarray}
Taking into account that $a=\sum_{l>k}b_{lk}$ we arrive at
\begin{equation}
P\left(\rho^{AB},H_{A}\right)=\max_{t}\left\{ \sum_{l>k}b_{lk}\cdot\left[1-\cos\left(\Delta E_{lk}t\right)\right]\right\} \;.\label{Ip2}
\end{equation}
Since $P\left(\rho^{AB},H_{A}\right)\geq\sum_{l>k}b_{lk}\cdot\left[1-\cos\left(\Delta E_{lk}t\right)\right]$
for all times $t\neq t_{max}$, it follows that $P\left(\rho^{AB},H_{A}\right)\geq2\cdot\max_{l>k}b_{lk}$.
Using the fact that $a=\sum_{l>k}b_{lk}\le N\max_{l>k}b_{lk}$ we
obtain that $\max_{l>k}b_{lk}\geq\frac{1}{N}\sum_{l>k}b_{lk}=\frac{a}{N}$,
where $N=(d_{A}-1)d_{A}/2$ is the number of different $b_{lk}$ terms.
Collecting these results and recalling the definition of the geometric
measure of discord $D_{A}^{\left(2\right)}\left(\rho^{AB}\right)$,
we find that the impact power of any nondegenerate, finite-dimensional
local Hamiltonian $H_{A}$ is bounded from below by a simple linear
function of the geometric measure of discord:
\begin{equation}
P\left(\rho^{AB},H_{A}\right)
\geq\frac{4D_{A}^{\left(2\right)}\left(\rho^{AB}\right)}{d_{A}\left(d_{A}-1\right)}\;.\label{eq:P-3}
\end{equation}
From Eq.(\ref{eq:P-3}), in complete analogy with the qubit case,
it follows that if the initial state has vanishing quantum correlations,
there always exists at least one nontrivial local Hamiltonian $H_{A}$
with vanishing impact power. Therefore, a nonvanishing impact power
implies and quantifies a nonvanishing degree of quantumness, regardless
of the local Hilbert space dimension of party $A$.

It is worth noticing that while throughout this paper we have made
use of the Hilbert-Schmidt norm, we are by no means limited to this
choice. Similar conclusions hold as well for the trace distance, which
is directly related to the distinguishability of quantum states~\cite{Geometry2008}.
Indeed, given two density matrices $\rho$ and $\omega$, their squared
trace distance is $(\mathrm{Tr}[\sqrt{(\rho-\omega)^{2}}])^{2}=(\sum_{i}|\lambda_{i}|)^{2}$,
where the $\{\lambda_{i}\}$ are the eigenvalues of $(\rho-\omega)$.
This quantity is obviously always larger or equal than the squared
Hilbert-Schmidt distance $\mathrm{Tr}[(\rho-\omega)^{2}]=\sum_{i}\lambda_{i}^{2}$.
Therefore, an impact power gap for quantum correlated states exists
also in the case in which we replace the Hilbert-Schmidt distance
with the trace distance, and hence similar results can be obtained
also in this case. As the latter is monotonic under general stochastic
maps, this result is relevant in the the light of a recent observation~\cite{Piani2012}
that due to the fact that the Hilbert-Schmidt distance is not monotonic
under stochastic maps, some reversible operations on unmeasured subsystem
$B$ can change the value of the quantum correlations.

In conclusion, we have 
established that all the quantum correlated states of bipartite quantum
systems exhibit a nonvanishing impact power gap, i.e. a nonvanishing
minimal change under the action of {\em any} nontrivial local Hamiltonian.
On the contrary for every classically correlated state there exists
at least one particular nontrivial local unitary operation that leaves
the state unchanged. Starting from this observation we have quantified
this global change via the Hilbert-Schmidt distance, and showed that
the minimal distance achieved along the local time evolution is proportional
to the amount of quantum correlations quantified via the geometric
measure of discord. Moreover, for two-qubit systems at fixed global
purity, we have verified explicitly that Werner states maximize the
impact power gap and thus the amount of quantum correlations. We have
mainly used as measure of the effect of the local unitary operations
the Hilbert-Schmidt metrics; however, we have shown that similar results
can be obtained also using the trace distance. On the other hand, it is
expected that the detailed structure of the quantification of nonclassicality
and the characterization of maximally quantum-correlated states using the
formalism of least-perturbing local unitary operations will depend to some
extent on the choice of the metric inducing the distance between quantum
states. In this respect the choice of the Bures metric, which is at the same
time monotonic and Riemannian, seems to be the most appropriate one, also in
light of the fundamental operational meaning that stems from its intimate
relation with the Uhlmann fidelity. The general structure of distance-based
measures of quantumness associated to least-perturbing local unitary operations
defined via different norms (Bures, trace, and Hilbert-Schmidt) and their
detailed comparison are the subject of ongoing investigations and we hope to
report on them in the near future \cite{preparation}.

{\em Acknowledgements:} AS and DB acknowledge financial support
from DFG and ELES, while SMG, WR, and FI acknowledge financial support
from the EU STREP Projects HIP, Grant Agreement No. 221889, and iQIT,
Grant Agreement No. 270843.

\end{document}